\documentclass[11pt,onecolumn]{IEEEtran}
\usepackage{amsthm}
\usepackage{times,amssymb,amsmath,amsfonts,float,nicefrac,color,bbm,mathrsfs,caption,float}
\usepackage{algorithm,enumerate,multirow,caption,tikz,graphicx}
\usepackage[mathscr]{eucal}
\usepackage{sidecap}
\usepackage{algpseudocode}
\usepackage{verbatim}
\usepackage{epstopdf}
\usepackage{textcomp}
\usetikzlibrary{shapes,arrows}
\usepackage{stmaryrd}
\usepackage{mathabx}
\usepackage[noadjust]{cite}
\usepackage{booktabs}
\usepackage[normalem]{ulem}
\usepackage[top=1in,bottom=1in,left=1in,right=1in]{geometry}
\allowdisplaybreaks

\newcommand{\RN}[1]{%
  \textup{\expandafter{\romannumeral#1}}%
}

\newcommand\remove[1]{}

\allowdisplaybreaks
\usepackage{enumitem}
\setlist[enumerate]{leftmargin=*}

\newtheorem{theorem}{Theorem}
\newtheorem{definition}{Definition}

\newtheorem{lemma}{Lemma}

\newtheorem{construction}{Construction}

\newcommand{\cC}{\mathcal{C}}

\newcommand{\cF}{\mathcal{F}}

\newcommand{\cR}{\mathcal{R}}

\DeclareMathOperator{\lcm}{lcm}

\DeclareMathOperator{\ce}{ce}
\DeclareMathOperator{\co}{co}

\usepackage{tikz}
\usetikzlibrary{calc}

\begin{document}
\title{New constructions of cooperative MSR codes: Reducing node size to $\exp(O(n))$}

\author{\IEEEauthorblockN{Min Ye}}

\maketitle

 


\begin{abstract}
We consider the problem of multiple-node repair in distributed storage systems under the {\em cooperative model}, where the repair bandwidth includes the amount of data exchanged between {\em any} two different storage nodes. Recently, explicit constructions of MDS codes with optimal cooperative repair bandwidth for all possible parameters were given by Ye and Barg (IEEE Transactions on Information Theory, 2019). The node size (or sub-packetization) in this construction scales as $\exp(\Theta(n^h))$, where $h$ is the number of failed nodes and $n$ is the code length.

In this paper, we give new explicit constructions of optimal MDS codes for all possible parameters under the cooperative model, and the node size of our new constructions only scales as $\exp(O(n))$ for any number of failed nodes. Furthermore, it is known that any optimal MDS code under the cooperative model (including, in particular, our new code construction) also achieves optimal repair bandwidth under the {\em centralized model}, where the amount of data exchanged between failed nodes is {\em not} included in the repair bandwidth. We further show that the node size of our new construction is also much smaller than that of the best known MDS code constructions for the centralized model.
\end{abstract}

\section{Introduction}\label{sect:intro}
Maximum Distance Separable (MDS) codes are widely used in distributed storage systems since they provide the optimal trade-off between the fault tolerance and storage overhead. More precisely, a distributed storage system encoded by an $(n,k)$ MDS code can tolerate the failure of any $r:=n-k$ storage nodes.
In practice, the system will repair the failed nodes when there are only one or a few ($<r$) node failures. In \cite{Dimakis10}, Dimakis et. al. suggested a new measure for the efficiency of the repair procedure, namely, the {\em repair bandwidth}, defined as the amount of data communicated between the storage nodes when repairing failed nodes.
Dimakis et. al. \cite{Dimakis10} further established the lower bound on the repair bandwidth, known as the {\em cut-set bound}, and showed the existence of codes that achieve the cut-set bound. Such codes are called {\em regenerating codes}. An important subclass of regenerating codes is {\em minimum storage regenerating} (MSR) code, that is, MDS code that achieves the cut-set bound with equality.
Constructions of MSR codes were proposed in \cite{Rashmi11,Tamo13,Ye16,Sasid16,Ye16a,Tamo17RS}.

While originally the concept of MSR codes (or more generally, regenerating codes) was proposed for single node repair \cite{Dimakis10}, studies into MSR codes
have expanded into the task of repairing multiple erasures. Multiple-node repair was mainly studied under two models: One is the {\em cooperative model}, where the repair bandwidth includes the amount of data exchanged between {\em any} two different storage nodes \cite{Kermarrec11,Shum13,Li14,Shum16}. The other is the {\em centralized model}, where all the failed nodes are recreated in one location, and the amount of data exchanged between the failed nodes is {\em not} included in the repair bandwidth \cite{Cadambe13,Ye16,Rawat18,Wang17,Zorgui17,Tamo19,Zorgui18}.
The cut-set bounds on the repair bandwidth under these two models were derived in \cite{Shum13} and \cite{Cadambe13} respectively.

Most studies of MSR codes in the literature are concerned with {\em array codes}\footnote{See \cite{Shanmugam14,Guruswami16,Ye16b,Tamo17RS,Dau17,Dau18,Chowdhury17,Mardia19,Tamo19} for results on repairing scalar MDS codes, such as Reed-Solomon codes.} \cite{Blaum98}. In particular, an $(n,k,l)$ MDS array code over a finite field $F$ is formed of $k$ information nodes and $r=n-k$ parity nodes with the property that the contents of any $k$ out of $n$ nodes suffices to recover the codeword. Every node is a column vector in $F^l,$ where the parameter $l$ is called the node size or {\em sub-packetization}.
More precisely, let $\cC$ be an $(n,k,l)$ MDS array code over a finite field $F$ and let $C\in \cC$ be a codeword. 
We write $C$ as $(C_1,C_2,\dots,C_n)$,
where $C_i\in F^l, 1\le i\le n$ is the $i$th node of $C$. 
Let $\cF\subseteq [n], |\cF|=h$ and $\cR\subseteq[n]\backslash \cF, |\cR|=d$ be the sets of indices of the failed nodes and
the helper nodes, respectively, where we use the notation $[n]:=\{1,2,\dots,n\}$.
Define $N_{\co}(\cC,{\cF},{\cR})$ and $N_{\ce}(\cC,{\cF},{\cR})$ as the smallest number of symbols of the finite field $F$ one needs to download  in order to recover the failed  
nodes $\{C_i,i\in{\cF}\}$ from the helper nodes $\{C_j,j\in{\cR}\}$ under the cooperative model and the centralized model, respectively. The cut-set bounds read as follows:

\begin{theorem}[Cut-set bound \cite{Dimakis10,Cadambe13,Shum13,Ye19}] \label{def:csb}
Let $\cC$ be an $(n,k,l)$ MDS array code. For any two disjoint subsets ${\cF},{\cR}\subseteq[n]$ such that $|{\cF}|\le r$ and $|{\cR}|\ge k,$ we have the following inequalities:
\begin{align}
N_{\co}(\cC,{\cF},{\cR}) & \ge \frac{|{\cF}|(|{\cR}|+|{\cF}|-1)l}{|{\cF}|+|{\cR}|-k} , \label{eq:cutset} \\
N_{\ce}(\cC,{\cF},{\cR}) & \ge \frac{|{\cF}||{\cR}|l}{|{\cF}|+|{\cR}|-k} . \label{eq:csce}
\end{align}
\end{theorem}

We use $n$ to denote the code length and $k$ to denote the code dimension throughout the paper.
MSR codes for multiple-node repair are defined as follows:
\begin{definition}[$(h,d)$-MSR code]
We say that a code $\cC$ is an $(h,d)$-MSR code under the cooperative (resp., centralized) model if 
{\em (1)} $\cC$ is an MDS code; 
{\em (2)} for any choice of $\cF\subseteq [n]$ and $\cR\subseteq[n]\backslash \cF$ with $|\cF|=h$ and $|\cR|=d$,
\begin{align*}
& N_{\co}(\cC,{\cF},{\cR})  = \frac{|{\cF}|(|{\cR}|+|{\cF}|-1)l}{|{\cF}|+|{\cR}|-k} ,  \\
& \Big( \text{resp.,~~} N_{\ce}(\cC,{\cF},{\cR})  = \frac{|{\cF}||{\cR}|l}{|{\cF}|+|{\cR}|-k}
\Big) . 
\end{align*}
\end{definition}

In \cite{Ye19}, it was shown that the cooperative model is stronger than the centralized model. More precisely,
\begin{theorem}[{\cite[Theorem 2]{Ye19}}]\label{thm:st}
Let $\cC$ be an $(n,k,l)$ MDS array code and let ${\cF},{\cR}\subseteq[n]$ be two disjoint subsets such that $|{\cF}|\le r$ and $|{\cR}|\ge k.$ If
$$
N_{\co}(\cC,{\cF},{\cR}) = \frac{|{\cF}|(|{\cR}|+|{\cF}|-1)l}{|{\cF}|+|{\cR}|-k},
$$
then
$$
N_{\ce}(\cC,{\cF},{\cR}) = \frac{|{\cF}||{\cR}|l}{|{\cF}|+|{\cR}|-k}.
$$
\end{theorem}
Therefore, any $(h,d)$-MSR code under the cooperative model is also an $(h,d)$-MSR code under the centralized model.

In the rest of the paper we focus on the cooperative model. Unless stated otherwise, all the concepts and objects mentioned below, such as the repair bandwidth and the cut-set bound, implicitly assume this model.

Until the recent work of Ye and Barg \cite{Ye19}, constructions of $(h,d)$-MSR codes were known only for some special values of $h$ and $d$: Paper \cite{Shum13} constructed such codes for the (trivial) case $d=k$, and \cite{Li14} presented $(h,d)$-MSR codes for the cooperative repair of two erasures in the regime of low rate $k/n \le 1/2$ (more precisely, \cite{Li14} constructed $(2,d)$-MSR codes for any $n,k,d$ such that $2k-3\le d\le n-2$).
In \cite{Ye19}, Ye and Barg gave explicit constructions of $(h,d)$-MSR codes for all possible values of $n,k,h,d$, i.e., for all $n,k,h,d$ such that $2\le h \le n-d\le n-k-1$.
While the construction of \cite{Ye19} solves the open problem of constructing cooperative MSR codes for general parameters, the node size of this construction is $((h+d-k)(d-k)^{h-1})^{\binom n h}$. This quantity scales as $\exp(\Theta(n^h))$ if we fix $r:=n-k$ (the number of parity nodes) and let $n$ grow\footnote{Note that $h\le r$ and $h+d\le n$, so $d-k<h+d-k\le r$. Therefore, $(h+d-k)(d-k)^{h-1}$ is upper bounded by $r^r$, which is a constant. For small $h$, the exponent $\binom n h$ scales as $\Theta(n^h)$. Thus the node size scales as $\exp(\Theta(n^h))$.}. Since $h\ge 2$ for multiple-node repair, the node size of the construction in \cite{Ye19} is too large for practical applications.

In this paper, we give new explicit constructions of $(h,d)$-MSR codes for all possible values of $n,k,h,d$, i.e., for all $n,k,h,d$ such that $2\le h \le n-d\le n-k$.
The node size of this new construction is $(h+d-k)(d-k+1)^n$. It scales as $\exp(O(n))$ if we fix $r$ and let $n$ grow. This is much smaller compared to the $\exp(\Theta(n^h))$ scaling in the construction of \cite{Ye19}.
Our new codes can be constructed over any finite field of size no smaller than $(d-k+1)n$. This requirement on the field size is exactly the same as the construction in \cite{Ye19}.

Here we briefly describe the construction of $(h,d)$-MSR code in \cite{Ye19} and
explain how we reduce the node size in our new construction. 
In \cite{Ye19}, the authors started with a scalar MDS code. This code has node size $l=1$ but only admits the naive/trivial repair method. Then paper \cite{Ye19} proposed a ``transform" on MDS codes with the following property:  Each time we apply this transform to an MDS code, we will obtain another MDS code that achieves the cut-set bound \eqref{eq:cutset} for the repair of one more $h$-tuple of failed nodes at the price of increasing the node size by a factor of $(h+d-k)(d-k)^{h-1}$ compared to the original MDS code.
For example, if we apply this transform to the scalar MDS code with node size $1$, then we will obtain an MDS code $\cC_1$ that can repair one $h$-tuple of failed nodes with optimal bandwidth (achieving the cut-set bound), and the node size of $\cC_1$ is $(h+d-k)(d-k)^{h-1}$. If we apply this transform again to the code $\cC_1$, then we will obtain an MDS code $\cC_2$ that can repair two $h$-tuples of failed nodes with optimal bandwidth, and the node size of $\cC_2$ is $((h+d-k)(d-k)^{h-1})^2$.
In total there are $\binom{n}{h}$ choices of $h$-tuples. Therefore, after applying this transform $\binom{n}{h}$ times, we will obtain an MDS code that achieves the cut-set bound for the repair of any $h$-tuple of failed nodes, and the node size of this code is $((h+d-k)(d-k)^{h-1})^{\binom n h}$. This is the $(h,d)$-MSR code constructed in \cite{Ye19}, and we can see that the exponent $\binom{n}{h}=\Theta(n^h)$ is inherent in this construction.
Therefore, we take a different approach in our new constructions. Quite surprisingly, we find that $(h,d)$-MSR codes can be obtained by replicating the $(1,d)$-MSR code construction in \cite{Ye16} $(h+d-k)$ times. The node size of the $(1,d)$-MSR code construction in \cite{Ye16} is $(d-k+1)^n$, so the node size of our new code construction is $(h+d-k)(d-k+1)^n$. The main novelty of this paper lies in the careful design of a new repair scheme that allows our new code construction to achieve the cut-set bound.

As mentioned above, any $(h,d)$-MSR code under the cooperative model is also an $(h,d)$-MSR code under the centralized model. Therefore, our new code also achieves the optimal repair bandwidth under the centralized model. The node size of the centralized $(h,d)$-MSR code construction in \cite{Ye16} is\footnote{In \cite{Ye16}, the authors did not explicitly present the construction of centralized $(h,d)$-MSR codes for a specific pair of $(h,d)$. Instead, they gave a universal code construction with the $(h,d)$-MSR property for all $h\le r$ and all $k\le d\le n-h$ simultaneously. However, from the proof of Theorem~11 in \cite{Ye16}, it is clear that for a specific pair of $(h,d)$ one can construct a centralized $(h,d)$-MSR code with node size $(\lcm(d+1-k,d+2-k,\dots,d+h-k))^n$.}
$$
(\lcm(d+1-k,d+2-k,\dots,d+h-k))^n,
$$
where $\lcm$ stands for least common multiple.
Since $\lcm(d+1-k,d+2-k,\dots,d+h-k)\ge (d+1-k)(d+2-k)$, the node size of the construction in \cite{Ye16} is at least $(d+1-k)^n (d+2-k)^n$, and this is much larger\footnote{We always have $(d+2-k)^n\gg n>h+d-k$, so $(d+1-k)^n (d+2-k)^n \gg (h+d-k)(d-k+1)^n$.} than $(h+d-k)(d-k+1)^n$, the node size of our new code construction.
Therefore, our new construction significantly reduces the node size even compared to previous constructions of centralized MSR codes.
In Table~\ref{table:parameters}, we list the node size of $(h,d)$-MSR codes constructed in \cite{Ye16,Ye19} and this paper.

\begin{table}[ht]
\captionsetup{width=.8\linewidth,font=scriptsize}
\centering
\begin{tabular}{|c|c|c|c|}
\hline
  & repair model & node size & scaling of node size \\
\hline
Ye-Barg 2017 \cite{Ye16} & centralized & $(\lcm(d-k+1,\dots,d-k+h))^n$ & $\exp(O(n))$ \\
 \hline
Ye-Barg 2019 \cite{Ye19} & both cooperative and centralized & $((h+d-k)(d-k)^{h-1})^{\binom n h}$ & $\exp(\Theta(n^h))$ \\
 \hline
This paper & both cooperative and centralized & $(h+d-k)(d-k+1)^n$ & $\exp(O(n))$ \\
 \hline
 \end{tabular}
 \caption{Node size of different constructions of $(h,d)$-MSR codes under the cooperative model and the centralized model. In the column ``repair model", we specify under which model(s) the code construction achieves the optimal repair bandwidth.
 In the column ``scaling of node size", we consider the regime of fixed $r$ and growing $n$.}\label{table:parameters}
\end{table}

The rest of this paper is organized as follows: We start with the simplest case of $h=2$ and $d=k+1$ in Section~\ref{sect:2k} to illustrate the main ideas of our new constructions. In Section~\ref{sect:2d}, we present the code construction for $h=2$ and any value of $d$. Finally, we deal with the most general case in Section~\ref{sect:hd}.

\section{Construction of $(2,k+1)$-MSR codes} \label{sect:2k}

In this section, we present the code construction and the corresponding repair scheme for the special case of $h=2$ and $d=k+1$.
More precisely, we construct an $(n,k,l=3\times 2^n)$ MDS array code $\cC$ together with a repair scheme that achieves the cut-set bound \eqref{eq:cutset} for the repair of any $2$ failed nodes from any $k+1$ helper nodes.
Let $C=(C_1,C_2,\dots,C_n)$ be a codeword of $\cC$, where each node $C_i$ is a vector of length $l=3\times 2^n$. For $1\le i\le n$, we write $C_i=(c_{i,b,a}:b\in\{1,2,3\},a\in\{0,1,\dots,2^n-1\})$.
For $a\in\{0,1,\dots,2^n-1\}$, we write its $n$-digit binary expansion as $a=(a_1,a_2,\dots,a_n)$, where $a_i\in\{0,1\}$ for all $1\le i\le n$.
For $a\in\{0,1,\dots,2^n-1\},i\in[n],u\in\{0,1\}$, we further define $a(i,u):=(a_1,\dots,a_{i-1},u,a_{i+1},\dots,a_n)$, i.e., $a(i,u)$ is obtained by replacing the $i$th digit of $a$ with $u$. In particular, $a(i,a_i\oplus 1)$ is obtained by flipping the $i$th digit of $a$, where $\oplus$ denotes addition over the binary field.
Now we are ready to present our code construction.

\begin{construction} \label{cons:1}
Let $F$ be a finite field of size $|F|\ge 2n$. Let $\{\lambda_{i,j}:i\in[n],j\in\{0,1\}\}$ be $2n$ distinct elements of $F$.
The code $\cC$ is defined by the following parity check equations:
\begin{equation} \label{eq:c1}
\sum_{i=1}^n \lambda_{i,a_i}^t c_{i,b,a} = 0 \text{~~~for all~} t\in\{0,1,\dots,n-k-1\}, ~b\in\{1,2,3\}, ~a\in\{0,1,\dots,2^n-1\} .
\end{equation}
\end{construction}
By this definition, it is clear that for every $b\in\{1,2,3\}$ and every $a\in\{0,1,\dots,2^n-1\}$, the vector $(c_{1,b,a},c_{2,b,a},\dots,c_{n,b,a})$ forms an $(n,k)$ MDS code\footnote{More precisely, it is an $(n,k)$ Generalized Reed-Solomon code with evaluation points $(\lambda_{1,a_1},\lambda_{2,a_2},\dots,\lambda_{n,a_n})$.}.
Therefore, $\cC$ is an $(n,k,l=3\times 2^n)$ MDS array code.
Also observe that if we only vary the value of $a\in\{0,1,\dots,2^n-1\}$ and fix the value of $b$, then all the coordinates indexed by this fixed value of $b$ form the $(1,k+1)$-MSR code constructed in Section~IV of \cite{Ye16}. 
Therefore, our construction is obtained by replicating the construction in \cite{Ye16} three times (because $b$ takes three possible values).

Next we show how to repair any two failed nodes from any $k+1$ helper nodes with optimal repair bandwidth. Similarly to the repair scheme in \cite{Ye19}, we also divide the repair process into two rounds: In the first round, each failed node downloads $2^n$ symbols of $F$ from each of the $k+1$ helper nodes. After the first round, each failed node is able to recover $2\times 2^n$ coordinates of itself, and it also gathers some information about the other failed node. Then in the second round, each failed nodes downloads $2^n$ symbols of $F$ from the other failed nodes, after which both failed nodes are able to recover all their coordinates.
The total amount of data exchange in this repair process is $2(k+2)2^n$, meeting the cut-set bound \eqref{eq:cutset} with equality.

It is clear from the code construction \eqref{eq:c1} that every node in the code $\cC$ plays the same role. In other words, there is not a ``special" node in this code construction. Therefore, without loss of generality we assume that the set of indices of the failed nodes is $\cF=\{1,2\}$, i.e., the first two nodes fail. We will limit ourselves to this special case for the simplicity of notation, and the repair scheme for the general case of $\cF=\{i_1,i_2\}$ can be obtained by replacing the indices $1$ and $2$ with $i_1$ and $i_2$, respectively.

Let $\cR$ be the set of indices of the $k+1$ helper nodes.
In the first round of the repair process, the first node $C_1$ downloads
\begin{equation} \label{eq:d1}
\{c_{i,1,a}+c_{i,2,a(1,a_1\oplus 1)} : a\in\{0,1,\dots,2^n-1\}, i\in\cR\}
\end{equation}
from the helper nodes, and the second node $C_2$ downloads
\begin{equation} \label{eq:d2}
\{c_{i,1,a}+c_{i,3,a(2,a_2\oplus 1)} : a\in\{0,1,\dots,2^n-1\}, i\in\cR\}
\end{equation}
from the helper nodes.

\begin{lemma}
After downloading \eqref{eq:d1}, $C_1$ is able to recover
$$
\big\{c_{1,b,a}: a\in\{0,1,\dots,2^n-1\}, b\in\{1,2\} \big\} 
\bigcup \big\{c_{2,1,a}+c_{2,2,a(1,a_1\oplus 1)} : a\in\{0,1,\dots,2^n-1\} \big\} .
$$
After downloading \eqref{eq:d2}, $C_2$ is able to recover
$$
\big\{c_{2,b,a}: a\in\{0,1,\dots,2^n-1\}, b\in\{1,3\} \big\} 
\bigcup \big\{c_{1,1,a}+c_{1,3,a(2,a_2\oplus 1)} : a\in\{0,1,\dots,2^n-1\} \big\} .
$$
\end{lemma}
\begin{proof}
According to \eqref{eq:c1}, for every $a\in\{0,1,\dots,2^n-1\}$, we have
\begin{align*}
\sum_{i=1}^n \lambda_{i,a_i}^t c_{i,1,a} = 0
\text{~~and~~}
\lambda_{1,a_1\oplus 1}^t c_{1,2,a(1,a_1\oplus 1)}
+ \sum_{i=2}^n \lambda_{i,a_i}^t  & c_{i,2,a(1,a_1\oplus 1)} =0  \\
& \text{for all~} t\in\{0,1,\dots,n-k-1\} .
\end{align*}
Summing these two equations, we obtain that
\begin{align*}
\lambda_{1,a_i}^t c_{1,1,a}+ \lambda_{1,a_1\oplus 1}^t c_{1,2,a(1,a_1\oplus 1)}
+ \sum_{i=2}^n \lambda_{i,a_i}^t  (c_{i,1,a} + & c_{i,2,a(1,a_1\oplus 1)}) = 0 \\
& \text{for all~} t\in\{0,1,\dots,n-k-1\} .
\end{align*}
Therefore, the vector 
$$
(c_{1,1,a}, c_{1,2,a(1,a_1\oplus 1)}, c_{2,1,a} +  c_{2,2,a(1,a_1\oplus 1)},
c_{3,1,a} +  c_{3,2,a(1,a_1\oplus 1)}, \dots, c_{n,1,a} +  c_{n,2,a(1,a_1\oplus 1)})
$$
forms a Generalized Reed-Solomon (GRS) code of length $n+1$. Since there are $n-k$ parity check equations, the dimension of this GRS code is $(n+1)-(n-k)=k+1$.
Thus any $k+1$ coordinates suffice to recover the whole vector. In particular, since $|\cR|=k+1$, the elements in \eqref{eq:d1} allow $C_1$ to recover 
\begin{align*}
\big\{c_{1,1,a}: a\in\{0,1,\dots,2^n-1\} \big\} 
& \bigcup  \big\{c_{1,2,a(1,a_1\oplus 1)} : a\in\{0,1,\dots,2^n-1\} \big\}  \\
& \bigcup \big\{c_{2,1,a}+c_{2,2,a(1,a_1\oplus 1)} : a\in\{0,1,\dots,2^n-1\} \big\}.
\end{align*}
Finally, we conclude the proof of the first part of this lemma by noticing that 
$$
\big\{c_{1,2,a(1,a_1\oplus 1)} : a\in\{0,1,\dots,2^n-1\} \big\} =
\big\{c_{1,2,a}: a\in\{0,1,\dots,2^n-1\} \big\}  .
$$
The second part of this lemma can be proved in the same way, and we do not repeat it here.
\end{proof}

In the second round of the repair process, $C_1$ downloads
\begin{equation} \label{eq:ex1}
\big\{c_{1,1,a}+c_{1,3,a(2,a_2\oplus 1)} : a\in\{0,1,\dots,2^n-1\} \big\}
\end{equation}
from $C_2$, and $C_2$ downloads
$$
\big\{c_{2,1,a}+c_{2,2,a(1,a_1\oplus 1)} : a\in\{0,1,\dots,2^n-1\} \big\}
$$
from $C_1$.
Since $C_1$ already knows the values of the coordinates in 
$$
\big\{c_{1,1,a}: a\in\{0,1,\dots,2^n-1\} \big\} ,
$$
given the elements in \eqref{eq:ex1}, $C_1$ is able to further recover
$$
\big\{c_{1,3,a(2,a_2\oplus 1)} : a\in\{0,1,\dots,2^n-1\} \big\}
= \big\{c_{1,3,a}: a\in\{0,1,\dots,2^n-1\} \big\}.
$$
Similarly, we can show that $C_2$ is able to recover 
$$
\big\{c_{2,2,a}: a\in\{0,1,\dots,2^n-1\} \big\} 
$$
after the second round of the repair process. Therefore, we have shown that both $C_1$ and $C_2$ can indeed recover all their coordinates using our repair scheme.
This concludes the description of our repair scheme as well as the proof of the following theorem:
\begin{theorem}
The code $\cC$ given in Construction~\ref{cons:1} is a $(2,k+1)$-MSR code.
\end{theorem}

\section{Construction of $(2,d)$-MSR codes for general $d$}  \label{sect:2d}

In this section, we present the code construction and the corresponding repair scheme for $h=2$ and general values of $d$. Let $s:=d+1-k$.
We construct an $(n,k,l=(d+2-k)s^n)$ MDS array code $\cC$ together with a repair scheme that achieves the cut-set bound \eqref{eq:cutset} for the repair of any $2$ failed nodes from any $d$ helper nodes.
In this section, each node $C_i$ is a vector of length $l=(d+2-k)s^n$. For $1\le i\le n$, we write $C_i=(c_{i,b,a}:b\in\{1,2,\dots,d+2-k\},a\in\{0,1,\dots,s^n-1\})$.
For $a\in\{0,1,\dots,s^n-1\}$, we write its $n$-digit $s$-ary expansion as $a=(a_1,a_2,\dots,a_n)$, where $a_i\in\{0,1,\dots,s-1\}$ for all $1\le i\le n$.
For $a\in\{0,1,\dots,s^n-1\},i\in[n],u\in\{0,1,\dots,s-1\}$, we further define $a(i,u):=(a_1,\dots,a_{i-1},u,a_{i+1},\dots,a_n)$, i.e., $a(i,u)$ is obtained by replacing the $i$th digit of $a$ with $u$. 
Now we are ready to present our code construction.

\begin{construction} \label{cons:2}
Let $F$ be a finite field of size $|F|\ge s n$. Let $\{\lambda_{i,j}:i\in[n],j\in\{0,1,\dots,s-1\}\}$ be $sn$ distinct elements of $F$.
The code $\cC$ is defined by the following parity check equations:
$$
\sum_{i=1}^n \lambda_{i,a_i}^t c_{i,b,a} = 0 \text{~~~for all~} t\in\{0,1,\dots,n-k-1\}, ~b\in\{1,2,\dots,d+2-k\}, ~a\in\{0,1,\dots,s^n-1\} .
$$
\end{construction}
Similarly to the previous section, we can show that $\cC$ is an $(n,k,l=(d+2-k)s^n)$ MDS array code, and it is also obtained by replicating the construction in Section~IV of \cite{Ye16} $(d+2-k)$ times.

Our repair scheme is again divided into two rounds: In the first round, each failed node downloads $s^n$ symbols of $F$ from each of the $d$ helper nodes. After the first round, each failed node is able to recover $(d+1-k) s^n$ coordinates of itself, and it also gathers some information about the other failed node. Then in the second round, each failed nodes downloads $s^n$ symbols of $F$ from the other failed nodes, after which both failed nodes are able to recover all their coordinates.
The total amount of data exchange in this repair process is $2(d+1)s^n$, meeting the cut-set bound \eqref{eq:cutset} with equality.

Similarly to the previous section, we also assume that the set of failed nodes is $\cF=\{1,2\}$ for the simplicity of notation. The repair scheme for the general case of $\cF=\{i_1,i_2\}$ can be obtained by replacing the indices $1$ and $2$ with $i_1$ and $i_2$, respectively.

Let $\cR$ be the set of indices of the $d$ helper nodes.
In this section, we use $\oplus$ to denote addition modulo $s$.
In the first round of the repair process, the first node $C_1$ downloads
\begin{equation} \label{eq:d3}
\Big\{ \sum_{j=0}^{s-2} c_{i,j+1,a(1,a_1\oplus j)}
+ c_{i,s,a(1,a_1\oplus (s-1))}
 : a\in\{0,1,\dots,s^n-1\}, i\in\cR \Big\}
\end{equation}
from the helper nodes, and the second node $C_2$ downloads
\begin{equation} \label{eq:d4}
\Big\{ \sum_{j=0}^{s-2} c_{i,j+1,a(2,a_2\oplus j)}
+ c_{i,s+1,a(2,a_2\oplus (s-1))} : a\in\{0,1,\dots,s^n-1\}, i\in\cR \Big\}
\end{equation}
from the helper nodes.

\begin{lemma}
After downloading \eqref{eq:d3}, $C_1$ is able to recover
\begin{align*}
& \big\{c_{1,b,a}: a\in\{0,1,\dots,s^n-1\}, b\in \{1,2,\dots,s-1,s\} \big\}   \\
\bigcup & \big\{\sum_{j=0}^{s-2} c_{2,j+1,a(1,a_1\oplus j)}
+ c_{2,s,a(1,a_1\oplus (s-1))} : a\in\{0,1,\dots,s^n-1\} \big\} .
\end{align*}
After downloading \eqref{eq:d4}, $C_2$ is able to recover
\begin{align*}
& \big\{c_{2,b,a}: a\in\{0,1,\dots,s^n-1\}, b\in\{1,2,\dots,s-1,s+1\} \big\}  \\
\bigcup & \big\{ \sum_{j=0}^{s-2} c_{1,j+1,a(2,a_2\oplus j)}
+ c_{1,s+1,a(2,a_2\oplus (s-1))} : a\in\{0,1,\dots,s^n-1\} \big\} .
\end{align*}
\end{lemma}
\begin{proof}
We only prove the first part of this lemma, and the second part can be proved in the same way.
According to the parity check equations in Construction~\ref{cons:2}, for every $a\in\{0,1,\dots,s^n-1\}$, we have
\begin{align*}
\lambda_{1,a_1\oplus j}^t c_{1,j+1,a(1,a_1\oplus j)}
+ \sum_{i=2}^n \lambda_{i,a_i}^t  & c_{i,j+1,a(1,a_1\oplus j)} =0  \\
& \text{for all~~} j\in\{0,1,\dots,s-1\} \text{~~and all~~} t\in\{0,1,\dots,n-k-1\} .
\end{align*}
Summing these equations over $j\in\{0,1,\dots,s-1\}$, we obtain that
\begin{align*}
\sum_{j=0}^{s-1} \lambda_{1,a_1\oplus j}^t c_{1,j+1,a(1,a_1\oplus j)}
+ \sum_{i=2}^n \lambda_{i,a_i}^t \Big( \sum_{j=0}^{s-1} & c_{i,j+1,a(1,a_1\oplus j)} \Big)=0   \\
& \text{for all~~} t\in\{0,1,\dots,n-k-1\} .
\end{align*}
Therefore, the vector 
$$
\Big( \big(c_{1,j+1,a(1,a_1\oplus j)}: j\in\{0,1,\dots,s-1\} \big),
\big(\sum_{j=0}^{s-1}  c_{i,j+1,a(1,a_1\oplus j)} : i\in\{2,3,\dots,n\} \big) \Big)
$$
forms a GRS code of length $n+s-1=n+d-k$. Since there are $n-k$ parity check equations, the dimension of this GRS code is $(n+d-k)-(n-k)=d$.
Thus any $d$ coordinates suffice to recover the whole vector. In particular, since $|\cR|=d$, the elements in \eqref{eq:d3} allow $C_1$ to recover 
\begin{align*}
& \big\{c_{1,j+1,a(1,a_1\oplus j)}: a\in\{0,1,\dots,s^n-1\}, j\in\{0,1,\dots,s-1\} \big\}   \\
\bigcup & \big\{\sum_{j=0}^{s-1} c_{2,j+1,a(1,a_1\oplus j)} : a\in\{0,1,\dots,s^n-1\} \big\} .
\end{align*}
Finally, we conclude the proof of the first part of this lemma by noticing that for every $j\in\{0,1,\dots,s-1\}$,
$$
\big\{c_{1,j+1,a(1,a_1\oplus j)}: a\in\{0,1,\dots,s^n-1\} \big\} =
\big\{c_{1,j+1,a}: a\in\{0,1,\dots,s^n-1\} \big\}  .
$$
\end{proof}

In the second round of the repair process, $C_1$ downloads
\begin{equation} \label{eq:ex2}
\big\{ \sum_{j=0}^{s-2} c_{1,j+1,a(2,a_2\oplus j)}
+ c_{1,s+1,a(2,a_2\oplus (s-1))} : a\in\{0,1,\dots,s^n-1\} \big\}
\end{equation}
from $C_2$, and $C_2$ downloads
$$
\big\{\sum_{j=0}^{s-2} c_{2,j+1,a(1,a_1\oplus j)}
+ c_{2,s,a(1,a_1\oplus (s-1))} : a\in\{0,1,\dots,s^n-1\} \big\}
$$
from $C_1$.
Since $C_1$ already knows the values of the coordinates in 
$$
\big\{c_{1,b,a}: a\in\{0,1,\dots,s^n-1\}, b\in\{1,2,\dots,s-1\} \big\} ,
$$
given the elements in \eqref{eq:ex2}, $C_1$ is able to further recover
$$
\big\{ c_{1,s+1,a(2,a_2\oplus (s-1))} : a\in\{0,1,\dots,s^n-1\} \big\}
= \big\{c_{1,s+1,a}: a\in\{0,1,\dots,s^n-1\} \big\}.
$$
Similarly, we can show that $C_2$ is able to recover 
$$
\big\{c_{2,s,a}: a\in\{0,1,\dots,s^n-1\} \big\} 
$$
after the second round of the repair process. Therefore, we have shown that both $C_1$ and $C_2$ can indeed recover all their coordinates using our repair scheme.
This concludes the description of our repair scheme as well as the proof of the following theorem:
\begin{theorem}
The code $\cC$ given in Construction~\ref{cons:2} is a $(2,d)$-MSR code.
\end{theorem}

\section{Construction of $(h,d)$-MSR codes for general $h$ and $d$} \label{sect:hd}

In this section, we deal with the most general cases. Let $s:=d+1-k$. For any $2\le h\le n-d\le n-k$, we construct an $(n,k,l=(d+h-k)s^n)$ MDS array code $\cC$ together with a repair scheme that achieves the cut-set bound \eqref{eq:cutset} for the repair of any $h$ failed nodes from any $d$ helper nodes.
In this section, each node $C_i$ is a vector of length $l=(d+h-k)s^n$. For $1\le i\le n$, we write $C_i=(c_{i,b,a}:b\in\{1,2,\dots,d+h-k\},a\in\{0,1,\dots,s^n-1\})$.
For $a\in\{0,1,\dots,s^n-1\}$, we write its $n$-digit $s$-ary expansion as $a=(a_1,a_2,\dots,a_n)$, where $a_i\in\{0,1,\dots,s-1\}$ for all $1\le i\le n$.
For $a\in\{0,1,\dots,s^n-1\},i\in[n],u\in\{0,1,\dots,s-1\}$, we further define $a(i,u):=(a_1,\dots,a_{i-1},u,a_{i+1},\dots,a_n)$, i.e., $a(i,u)$ is obtained by replacing the $i$th digit of $a$ with $u$. 
Now we are ready to present our code construction.

\begin{construction} \label{cons:3}
Let $F$ be a finite field of size $|F|\ge s n$. Let $\{\lambda_{i,j}:i\in[n],j\in\{0,1,\dots,s-1\}\}$ be $sn$ distinct elements of $F$.
The code $\cC$ is defined by the following parity check equations:
$$
\sum_{i=1}^n \lambda_{i,a_i}^t c_{i,b,a} = 0 \text{~~~for all~} t\in\{0,1,\dots,n-k-1\}, ~b\in\{1,2,\dots,d+h-k\}, ~a\in\{0,1,\dots,s^n-1\} .
$$
\end{construction}
Similarly to the previous section, we can show that $\cC$ is an $(n,k,l=(d+h-k)s^n)$ MDS array code, and it is also obtained by replicating the construction in Section~IV of \cite{Ye16} $(d+h-k)$ times.

Our repair scheme is again divided into two rounds: In the first round, each failed node downloads $s^n$ symbols of $F$ from each of the $d$ helper nodes. After the first round, each failed node is able to recover $(d+1-k) s^n$ coordinates of itself, and it also gathers some information about the other failed nodes. Then in the second round, each failed nodes downloads $s^n$ symbols of $F$ from each of the other $h-1$ failed nodes, after which all the failed nodes are able to recover all their coordinates.
The total amount of data exchange in this repair process is $h(d+h-1)s^n$, meeting the cut-set bound \eqref{eq:cutset} with equality.

Similarly to the previous section, we assume that the set of failed nodes is $\cF=\{1,2,\dots,h\}$ for the simplicity of notation. The repair scheme for the general case of $\cF=\{i_1,i_2,\dots,i_h\}$ can be obtained by replacing the index $j$  with $i_j$ for all $j\in[h]$.

Let $\cR\subseteq[n]\setminus\cF$ be the set of indices of the $d$ helper nodes.
In this section, we use $\oplus$ to denote addition modulo $s$.
In the first round of the repair process, the $u$th node $C_u, u\in[h]$ downloads
\begin{equation} \label{eq:d5}
\Big\{ \sum_{j=0}^{s-2} c_{i,j+1,a(u,a_u\oplus j)}
+ c_{i,s+u-1,a(u,a_u\oplus (s-1))}
 : a\in\{0,1,\dots,s^n-1\}, i\in\cR \Big\}
\end{equation}
from the helper nodes.

\begin{lemma}
After downloading \eqref{eq:d5}, the $u$th node $C_u, u\in[h]$ is able to recover
\begin{equation}  \label{eq:ttv}
\begin{aligned}
& \big\{c_{u,b,a}: a\in\{0,1,\dots,s^n-1\}, b\in \{1,2,\dots,s-1,s+u-1\} \big\}   \\
\bigcup & \big\{ \sum_{j=0}^{s-2} c_{i,j+1,a(u,a_u\oplus j)}
+ c_{i,s+u-1,a(u,a_u\oplus (s-1))}
 : a\in\{0,1,\dots,s^n-1\}, i\in[h]\setminus\{u\} \big\} .
\end{aligned}
\end{equation}
\end{lemma}
\begin{proof}
According to the parity check equations in Construction~\ref{cons:3}, for every $a\in\{0,1,\dots,s^n-1\}$ and every $t\in\{0,1,\dots,n-k-1\}$, we have
\begin{align*}
& \lambda_{u,a_u\oplus j}^t c_{u,j+1,a(u,a_u\oplus j)}
+ \sum_{i\neq u} \lambda_{i,a_i}^t   c_{i,j+1,a(u,a_u\oplus j)}  = 0  
 \text{~~for all~} j\in\{0,1,\dots,s-2\}   \\
& \lambda_{u,a_u\oplus (s-1)}^t c_{u,s+u-1,a(u,a_u\oplus (s-1))} + 
\sum_{i\neq u} \lambda_{i,a_i}^t  c_{i,s+u-1,a(u,a_u\oplus (s-1))}  = 0 .
\end{align*}
Summing these $s$ equations, we obtain that for every $a\in\{0,1,\dots,s^n-1\}$ and every $t\in\{0,1,\dots,n-k-1\}$,
\begin{align*}
\sum_{j=0}^{s-2} \lambda_{u,a_u\oplus j}^t c_{u,j+1,a(u,a_u\oplus j)}
& + \lambda_{u,a_u\oplus (s-1)}^t c_{u,s+u-1,a(u,a_u\oplus (s-1))} \\
& + \sum_{i\neq u} \lambda_{i,a_i}^t \Big( \sum_{j=0}^{s-2} c_{i,j+1,a(u,a_u\oplus j)}
+ c_{i,s+u-1,a(u,a_u\oplus (s-1))} \Big)  =0 .
\end{align*}
Therefore, the vector
\begin{align*}
\Big( \big( c_{u,j+1,a(u,a_u\oplus j)} : j\in\{0,1,\dots,s-2\}  \big)  , ~
& c_{u,s+u-1,a(u,a_u\oplus (s-1))}  ,  \\
& \big(\sum_{j=0}^{s-2} c_{i,j+1,a(u,a_u\oplus j)}
+ c_{i,s+u-1,a(u,a_u\oplus (s-1))} : i\in[n]\setminus\{u\}  \big)
\Big)
\end{align*}
forms a GRS code of length $n+s-1=n+d-k$. Since there are $n-k$ parity check equations, the dimension of this GRS code is $(n+d-k)-(n-k)=d$.
Thus any $d$ coordinates suffice to recover the whole vector. In particular, since $|\cR|=d$, the elements in \eqref{eq:d5} allow $C_u$ to recover 
\begin{align*}
& \big\{c_{u,j+1,a(u,a_u\oplus j)} : a\in\{0,1,\dots,s^n-1\}, j\in\{0,1,\dots,s-2\}  \big\}   \\
\bigcup & \big\{c_{u,s+u-1,a(u,a_u\oplus (s-1))} : a\in\{0,1,\dots,s^n-1\}  \big\} \\
\bigcup & \big\{ \sum_{j=0}^{s-2} c_{i,j+1,a(u,a_u\oplus j)}
+ c_{i,s+u-1,a(u,a_u\oplus (s-1))}
 : a\in\{0,1,\dots,s^n-1\}, i\in[h]\setminus\{u\} \big\} .
\end{align*}
Finally, we conclude the proof of this lemma by noticing that 
\begin{align*}
\big\{c_{u,j+1,a(u,a_u\oplus j)} : a\in\{0,1,\dots,s^n-1\} \big\} & =
\big\{c_{u,j+1,a}: a\in\{0,1,\dots,s^n-1\} \big\}  \\
& \hspace*{1.2in} \text{for every~}  j\in\{0,1,\dots,s-2\} , \\
\big\{c_{u,s+u-1,a(u,a_u\oplus (s-1))} : a\in\{0,1,\dots,s^n-1\}  \big\}
& = \big\{c_{u,s+u-1,a} : a\in\{0,1,\dots,s^n-1\}  \big\} .
\end{align*}
\end{proof}

In the second round of the repair process, the failed node $C_u,u\in[h]$ downloads
\begin{equation} \label{eq:ex3}
\big\{ \sum_{j=0}^{s-2} c_{u,j+1,a(i,a_i\oplus j)}
+ c_{u,s+i-1,a(i,a_i\oplus (s-1))}
 : a\in\{0,1,\dots,s^n-1\} \big\}
\end{equation}
from each of the other failed nodes $C_i,i\in[h]\setminus\{u\}$.
Since $C_u$ already knows the values of the coordinates in 
$$
\big\{c_{u,b,a}: a\in\{0,1,\dots,s^n-1\}, b\in\{1,2,\dots,s-1\} \big\} ,
$$
given the elements in \eqref{eq:ex3}, $C_u$ is able to further recover
$$
\big\{  c_{u,s+i-1,a(i,a_i\oplus (s-1))}
 : a\in\{0,1,\dots,s^n-1\} \big\}
 = \big\{  c_{u,s+i-1,a}
 : a\in\{0,1,\dots,s^n-1\} \big\}
$$
for all $i\in[h]\setminus\{u\}$.
Combining this with \eqref{eq:ttv}, we conclude that after the second round of repair, each failed node $C_u$ is able to recover
\begin{align*}
& \big\{c_{u,b,a}: a\in\{0,1,\dots,s^n-1\}, b\in \{1,2,\dots,s-1,s+u-1\} \big\}   \\
 & \bigcup \big\{  c_{u,s+i-1,a}
 : a\in\{0,1,\dots,s^n-1\}, i\in[h]\setminus\{u\} \big\}  \\
 = & \big\{c_{u,b,a}: a\in\{0,1,\dots,s^n-1\}, b\in \{1,2,\dots,s+h-1\} \big\},
\end{align*}
i.e., each failed node is able to recover all its coordinates.
This concludes the description of our repair scheme and proves the following theorem:
\begin{theorem}
The code $\cC$ given in Construction~\ref{cons:3} is a $(h,d)$-MSR code.
\end{theorem}

\bibliographystyle{IEEEtran}
\bibliography{repair}

\end{document}